\documentclass[11pt]{article}
\usepackage[margin=1.2in]{geometry}
\usepackage{amsmath,amssymb,amsfonts,amsthm}
\usepackage{parskip}
\usepackage{breakcites}
\usepackage{graphbox}

\usepackage{algorithm}
\usepackage{algpseudocode}
\usepackage{float}
\usepackage{booktabs}
\usepackage{lipsum}
\usepackage{xcolor}
\usepackage{xspace}
\usepackage{authblk}
\usepackage{wrapfig}
\usepackage{subcaption}
\usepackage{thm-restate}
\newcommand{\argmax}{\operatornamewithlimits{argmax}}
\usepackage{breakcites}
\usepackage{bbm}

\usepackage{bm}
\definecolor{codegreen}{rgb}{0,0.9,0}
\definecolor{codegray}{rgb}{0.5,0.5,0.5}
\definecolor{codepurple}{rgb}{0.58,0,0.82}
\definecolor{backcolour}{rgb}{0.1,0.1,0.1} 
\definecolor{codecolor}{rgb}{0,0.7,0} 

\usepackage[numbers,compress]{natbib}
\usepackage[pagebackref=true]{hyperref}
\usepackage{xcolor}
\definecolor{myblue}{rgb}{0.36, 0.54, 0.66}
\hypersetup{
	colorlinks=true,
	linkcolor=magenta,
	citecolor=myblue
}

\newcommand{\nrcomment}[1]{{\color{blue} Naveen: {#1}}}

\theoremstyle{plain}
\newtheorem{assumption}{Assumption}
\newtheorem{definition}{Definition}
\newtheorem{theorem}{Theorem}

\usepackage{graphicx}

\newcommand\E{\mathbb{E}}

\usepackage{tocbasic}
\DeclareTOCStyleEntry[
  beforeskip=.05em plus .6pt,
  pagenumberformat=\textbf
]{tocline}{section}

\newif\ifarxiv
\arxivtrue

\title{Data-driven Design of Randomized Control Trials with Guaranteed Treatment Effects}

\author[1]{Santiago Cortes-Gomez}
\author[1]{Naveen Raman}
\author[1]{Aarti Singh}
\author[1]{Bryan Wilder}

\affil[1]{Machine Learning Department, Carnegie Mellon University}
\affil[1]{\texttt{\{scortesg, naveenr, aartisingh, bwilder\}@cs.cmu.edu}}

\begin{document}

 \maketitle

\begin{abstract}
    Randomized controlled trials (RCTs) can be used to generate guarantees on treatment effects. 
However, RCTs often spend unnecessary resources exploring sub-optimal treatments, which can reduce the power of treatment guarantees. To address these concerns, we develop a two-stage RCT  where, first on a data-driven screening stage, we prune low-impact treatments, while in the second stage, we develop high probability lower bounds on the treatment effect. 
Unlike existing adaptive RCT frameworks, our method is simple enough to be implemented in scenarios with limited adaptivity.
We derive optimal designs for two-stage RCTs and demonstrate how we can implement such designs through sample splitting.
Empirically, we demonstrate that two-stage designs improve upon single-stage approaches, especially in scenarios where domain knowledge is available in the form of a prior. Our work is thus, a simple, yet effective, method to estimate high probablility certificates for high performant treatment effects on a RCT. 

\end{abstract}

\section{Introduction}

Randomized controlled trials (RCTs) are the gold standard for measuring treatment effects ~\citep{rct}. In a traditional single-stage RCT, the experimenter fixes a set of treatments up front and randomizes the samples across them using predetermined assignment probabilities. While this design simplifies implementation and analysis, it often spends samples exploring suboptimal arms. Since underperforming arms are unlikely to be implemented, precisely quantifying their effect sizes at the cost of reducing the statistical power for higher-quality treatments is not worthwhile. Policy decisions on scaling or adopting programs often depend on proving sufficient effectiveness, yet single-stage trials cannot allocate more samples to high-performing arms to improve the estimate precision.

Motivated by these limitations, there has been a growing body of work on adaptive trials, which periodically update the probability of assigning units to each arm in order to focus on more promising arms. While adaptive designs can enhance statistical performance, their implementation in practice can be challenging. Treatment assignment probabilities must be updated dynamically, potentially after each sample, which imposes significant logistical demands on practitioners conducting the trial. Furthermore, high degrees of adaptivity may even be infeasible when outcomes are delayed. For instance, in clinical trials for chronic disease treatments, effects may take months or even years to observe, making it impossible to repeatedly change assignment probabilities based on realized outcomes. 

In this work, we question whether complex, adaptive designs are needed to capture most of the value of improved statistical performance. We develop two-stage designs which operate within a deliberately restricted framework. The experimenter uniformly randomizes over all arms in a first stage, uses the results to pick a subset of arms to retain for the second stage of the trial, and then uniformly randomizes once more over the arms that are retained. That is, the algorithm outputs both its guess for a high performing arm \textit{and} a statistically valid lower bound for the arm's mean outcome. The goal is to return as large a lower bound as possible, which requires both identifying an arm with high reward and concentrating enough samples on that arm to quantify (\textit{certify}) its outcome accurately. Notably, the only decision is which arms to keep in the second stage, thus aiming to strike a balance between simplicity and performance. 

We make three main contributions. First, we design a novel algorithm to approximate the optimal two-stage design, with the goal of identifying a high probability lower bound. We theoretically analyze this algorithm's performance and prove an approximation guarantee relative to the optimal policy. Second, we extend our formulation to the Bayesian setting where the experimenter has a prior over arm means. We also present an algorithms and accompanying approximation guarantees for this setting. Third, we empirically demonstrate (code available at [hidden]) that our two-stage algorithms outperform single-stage designs on both synthetic and real-world datasets, without the logistical burden of complex adaptive designs.

\subsection{Related work}

Naturally, our work is related to best-arm identification \citep{jamieson2014best, auer2002using} and top-K arm identification \citep{bubeck2013multiple}. 
However, our approach differs from these settings in two key ways: first, our initial screening for good arms may not include all of the best ones for a given size, and it may not even contain the best arm. 
This is because our method is not designed for exhaustive exploration if it sacrifices certifying the final estimated effect. An interesting comparison can be made though with finding the $\epsilon$ best arms \citep{mason2020finding} setup. These formulations output all arms within an $\epsilon$ distance of the best one. However, once $\epsilon$ is fixed, the formulation remains exhaustive, as it aims to identify all good arms determined by the hyperparameter. In contrast, our method is not designed to optimize for exhaustive exploration in this way. Perhaps closest to our work  is \citet{katz2020true} where they study the sample complexity inherent of $\epsilon$ best arms and of finding a fixed number of arms larger than a threshold. The second problem is somewhat related to our setting as they don't care about finding all good arms above the given threshold, but instead just a subset of such arms.  However, they still fix an absolute hyperparameter, which is the number of good arms they aim to identify (in addition to a threshold defining the good arms). In contrast, our work adaptively chooses the number of arms to generate the largest certificate. 

A current topic in adaptive trials literature is the study of the two usually desired, yet conflicting goals, for adaptive designs. One is to optimize for the cumulative regret, i.e., the benefit to the participants {\em during} the experiment. The other is to optimize for the information gained via the experiment for selection and deployment of the best arm {\em after} the experiment. The latter can be quantified via, e.g., the simple regret, of the precision with which treatment effects can be estimated. Previous work \citep{bubeck2011pure} shows that these goals are irreconcilable; designs with lower cumulative regret trade off nearly one-for-one in the worst case with information gain objectives . Our focus in this paper is on designs which identify a high-performing arm with the strongest statistical power possible for future deployment (i.e., our aim is not to minimize cumulative regret). \citep{athey2022contextual,li2010contextual,bastani2021efficient,kasy2021adaptive,deshmukh2018simple,chambaz2017targeted, simchi2023multi, fan2021fragility}.

Finally, identifying estimands with statistically valid properties after an adaptive trial, such as our \textit{certification}, remains an active area of research. Due to the correlations introduced by adaptive procedures, performing inference on estimands obtained after the trial is concluded, without having a follow-up independent data collection for the identified arms, is not straightforward. Solutions to this problem have emerged from safe anytime inference \citep{waudby2021time} or by imposing additional conditions, as in batch-only inference \citep{chen2023optimal}. Our goal, instead, is to rely on the standard uniform allocation RCT, that is widely accepted by practitioners, in each stage to generate a certificate.
\section{Introducing Two-Stage RCTs}

We introduce the problem of finding a good certificate for a randomized control trial (RCT). 

\subsection{Problem Formulation}

Many real-world scenarios are more concerned in producing a high probability lower bound on the impact of a treatment rather than guaranteeing if it is optimal or not. For example, in policy settings, practitioners aim to give a guarantee on the performance of a policy (e.g. the policy has a certain positive treatment effect) in order to justify a course of action dictated by it. Formally, consider a set of $n$ arms, with means $\bm{\mu} = (\mu_{1},\mu_{2},\ldots,\mu_{n})$ and distributions $D_{\mu_i}$. Here, each $i$ corresponds to a treatment whose effect we determine through an RCT. Our objective is to produce a high-probability lower bound $l$ for one arm $\mu_i$, aiming to maximize $l$ to ensure a \textit{certified} high effect for arm $i$.  We define a certificate as follows: 

\begin{definition}
\textbf{Certificate} - Let $l$ be an estimand such that $l \leq \mu_i$ with probability $1 - \delta$ for some $i \in [n]$. We will say that $l$ is a certificate for $\mu_i$.
\end{definition}

Naturally, we can compute certificates in a single-stage RCT by uniformly allocating data to each arm. Nevertheless,  as we mentioned in the introduction, this approach wastes samples on unpromising arms, thereby increasing the variance on the estimand $l$. On the other extreme, fully adaptive trials are less than ideal as RCTs are typically run in-batch and with fixed budgets across many individuals, making them logistically cumbersome or, in the worst case, infeasible to deploy. Therefore, we propose a deliberately simplified two-stage RCT that captures the benefits of adaptivity without introducing unnecessary complexity. In our approach, the first stage filters out suboptimal arms using a policy $\pi$, while the second stage computes the certificate $l$ by allocating more interventions to the high-performing arms, resulting in a more precise estimate of the certificate $l$.

More precisely, let $T$ be our total budget, $s_{1}$ be our budget in the first stage and $s_{2}$ be our budget in the second stage, with $s_{1} + s_{2} = T$.  Let $X_1,X_2,\ldots,X_{s_{1}}$ be the data set for a first stage and $Y_1,Y_2,\ldots,Y_{s_2}$ the data for the second stage. Let $k$ be the number of arms selected for the second stage. For both the first and second stage, we uniformly explore all arms, e.g. in the first stage, we explore each arm $\lfloor \frac{s_{1}}{n} \rfloor$ times.

We determine the arms which survive from the first to the second stage through a function, $\pi(X_1,X_2,\ldots,X_{s_{1}})$.
This function maps first-stage results to a set of surviving arms, $\pi(X_1,X_2,\ldots,X_{s_{1}}) \subseteq [n]$. To avoid overloading notation, let $\mathbf{X} = \pi(X_1,X_2,\ldots,X_{s_{1}})$ and $\mathbf{Y} = Y_1,Y_2,\ldots,Y_{s_2}$. 
Our objective is then 
\begin{equation}
    \small
    f_{\bm{\mu}}(\pi) = \E_{\mathbf{X},\mathbf{Y} \sim \bm{\mu}}\left[\max_{i \in \pi(\mathbf{X})}l_i(\mathbf{Y})\right]
    \label{objective}
\end{equation}
The randomness behind choosing $\pi(\mathbf{X})$ is driven by the first-stage data, $\mathbf{X}$, while the randomness from the bound $l_i(\mathbf{Y})$, is from the second-stage, $\mathbf{Y}$. 

From our proposed objective, is clear that we focus on neither finding the single best arm nor any subset of the $k$-best arms, but rather on obtaining the best possible policy $\pi$ that, on average, will lead to the highest certificate. 
Intuitively, $\pi(\mathbf{X})$ is obtained by quickly screening the arms on the first stage discarding less promising arms using minimal data. This approach frees up more budget to be allocated to a subset of arms that show genuine potential, thereby improving the quality of the final certificate.






\subsection{Certificate estimation} 

Naturally objective \ref{objective} will be a function of the way the certificate $l$ is estimated. With the aim of outlining a concrete proposal to compute $l$, we begin by conditioning on a draw of the first stage. In other words, fix the $\pi(\mathbf{X})$, which would be a set of size $k$. Let $\overline{Y}_i$ denote the empirical mean from the second stage for arm $i$.
Assuming our the rewards are bounded, we show that $\overline{Y}_i$ concentrates near $\mu_{i}$ (using Hoeffding's inequality): 
\begin{equation}
    \label{certificate}
    \small
    P\left(|\overline{Y}_i - \mu_i| < \sqrt{\frac{k}{2 s_2}log(\frac{2}{\delta})}\right) \geq 1 - \delta
\end{equation}
Our certficiate is therefore $l = \overline{Y}_i -\sqrt{\frac{k}{2 s_2}log(\frac{2}{\delta})}$. 

We note that there are two competing terms when trying to maximize $l$ in Equation~\ref{certificate}. Intuitively, a certificate is maximized when a large set of arms is pruned, as this minimizes $k$, while still leaving sufficient budget for accurately computing $l$ for each arm. 

More formally, analyzing the latter term of the certificate, it can be seen that a two-stage approach provides an advantage over a single-stage  approach if: 
\begin{equation*}
    \small
    \sqrt{\frac{1}{2|s_2/k|} \log\left(\frac{2}{\delta}\right)} < \sqrt{\frac{1}{2|T/n|} \log\left(\frac{2}{\delta}\right)}    
\end{equation*}

Such a result occurs because single-stage methods can use all $T/n$ samples per arm due to its lack of adaptivity, while we only use $s_{2}/k$. 
For such a comparison to hold, we need that to not discard the best arm. 
Therefore, we select $k$ to ensure that Equation~\ref{certificate} holds, while still capturing an arm with large $\mu_{i}$. 

\section{Finding optimal policies for Two-Stage RCTs}


The key challenge in constructing two-stage designs is finding the  $\pi$ which maximizes Equation~\ref{objective}. This is the case due to the combinatorial nature of the problem. Therefore,  we will start by focusing on a simplified class of policies for clarity and ease of analysis, the so-called \textit{top}-$k$ policies. We will then show that, under a natural condition on the reward distributions, the class of top-$k$ policies contains the optimal policy.

\subsection{Designing Top-K Policies}
We begin by defining top-k policies. 
\begin{definition}
    \textbf{Top-K Policy} - Let $\sigma(\mathbf{X})$ be the descending ordering of arms by empirical mean from the first stage; that is $\bar{X}_{\sigma(\mathbf{X})_{1}} \geq \bar{X}_{\sigma(\mathbf{X})_{2}},\ldots,\bar{X}_{\sigma(\mathbf{X})_{n}}$, where ties are broken randomly. 
A top-k policy outputs sets of the form $\pi(\mathbf{X}) = \{\sigma(\mathbf{X})_1, \sigma(\mathbf{X})_2...\sigma(\mathbf{X})_{k(\mathbf{X})}\}$ for some function $k(\mathbf{X})$. 
\end{definition}

We analyze the performance of top-k policies under a stochastic dominance condition: 
\begin{definition}
    \textbf{First-order stochastic dominance} A random variable $A$ is said to first order stochastically dominate a random variable $B$ if  $P(B \geq x) \leq P(A \geq x)$ for all $x$.
\end{definition}
\begin{assumption}
\label{assumption_1}
For any $i,j$ such that $\mu_i \geq \mu_j$, $D_{\mu_1}$ first-order stochastically dominates $D_{\mu_j}$
\end{assumption}
Such an assumption is necessary to allow for comparison between the distribution of rewards for arms based on the empirical mean; without such an assumption comparing arms using only the empirical means would be difficult. 
This formulation encompasses a wide variety of set-ups. For example, this assumption holds if all treatment arms have binary outcomes, a common form of outcomes in practice, or if outcomes are normally distributed with the same variance. 

We aim to show that when outcome distributions respect stochastic domination, the optimal policy for Equation~\ref{objective} can be written as a top-k policy. 
The first step in doing is to demonstrate that stochastic domination in the true also applies to domination in empirical mean. 
\begin{restatable}{lemma2}{stochasticdomination}
\label{thm:stochasticdomination}
    Let $\sigma$ be the descending ordering of arms by empirical mean observed on the first stage. Then, for any $i < j$, $D_{\mu_{\sigma_i(\mathbf{X})}}$ first-order stochastically dominates, $D_{\mu_{\sigma_i(\mathbf{X})}}$ 

\end{restatable}
This allows us to translate any policy into a top-k counterpart which achieves higher value for $f(\pi)$. 
\begin{restatable}{lemma2}{greatertopk}
\label{thm:top_k_dominance}
    Fix an arbitrary policy $\pi$. Define it's top-k counterpart $\pi'$ to be the top-k policy with $k(\mathbf{X}) = |\pi(\mathbf{X})|$. That is, it selects the same number of arms as $\pi$ does for every realization of the trial, but it selects those arms to be the ones with largest empirical mean. Then $f_{\bm{\mu}}(\pi') \geq f_{\bm{\mu}}(\pi)$.    
\end{restatable}
The main idea is as follows: consider any set of arms $\pi(\mathbf{X})$ selected by the original policy. If there is an arm outside of $\pi(\mathbf{X})$ with strictly higher empirical mean than contained within $\pi(\mathbf{X})$, we can only do better by swapping in the arm with higher empirical mean. Formally, this follows because stochastic dominance also implies a corresponding ordering on any monotone function of a random variable, including the $\max$ which appears in our objective function. 

Together, these results imply that the existance of an optimal top-k policy:
\begin{restatable}{theorem}{topk}
\label{thm:topk}
  Let $\pi^*$ the policy that maximizes $f_{\bm{\mu}}$. There exists a top-k policy $\pi$ such that $f_{\bm{\mu}}(\pi^*)= f_{\bm{\mu}}(\pi)$. 
\end{restatable}

\subsection{Sample Splitting}
While Theorem~\ref{thm:topk} demonstrates that the optimal policy is top-k, the optimal $k$, which we denote $k^{*}$, is unknown.
We propose a sample splitting algorithm to estimate $k^{*}$. 
We construct our sample splitting algorithm by splitting the data from the first stage into two halves: training ($U$) and  validation ($V$). 
Our training half is used to compute empirical means and sort arms. 
Our validation half is used to compute certificate values for each value of $k$; essentially, we use the validation set $V$ to simulate the second-stage. 
We can then estimate the certificate value for different choices of $k$, and is analogous to train-validation splits for hyperparameter selection. 
We provide pseudocode below:
\begin{algorithm}
\caption{Sample splitting algorithm}
\begin{algorithmic}[1]
\State \textbf{Input:} $s_{1}$ iid samples.
\State \textbf{Output:} Set $\pi(\mathbf{X})$
\State Split first stage data randomly into two sets: $ U = \{x_1,...,x_{\frac{s_1}{2}}\}$ and $V=\{z_1,...,z_{\frac{s_1}{2}}\}$.
\State Compute $\bar{U}$, which is the average per-arm using data from $U$
\State Let $\sigma$ be the ordering of arms according to $\bar{U}$
\State Compute $\bar{V}$, which is the average per-arm using data from $V$
\For{$i = 1$ to k} 
    \State set $l_i = \argmax_{j \in [i]} \bar{V}_{\sigma_{j}} - \sqrt{\frac{i}{2s_2}log(\frac{2}{\delta})}$
\EndFor
\State $k = \argmax_{i \in [n]} l_i$
\State Let $\sigma'(\mathbf{X})$ be ordering of arms according $U \cup V$
\State \Return $\pi(\mathbf{X}) = \{\sigma'(\mathbf{X})_{1},\sigma'(\mathbf{X})_{2},\ldots,\sigma'(\mathbf{X})_{k}\}$
\end{algorithmic}
\label{algo_1}
\end{algorithm}

To bound the performance of our sample splitting algorithm, we compare against an optimal policy which selects the optimal number of arms, $k^{*}$. 
We lower bound the performance of our sample splitting algorithm against such an optimal algorithm: 

\begin{restatable}{proposition}{expectedvaluecertificate}
\label{thm:expectedvaluecertificate}
Let $\bm{\mu}$ be such that assumption \ref{assumption_1} holds. Let $\sigma$ be the permutation of the indexes obtained from sorting the empirical means obtained on the first stage on descending order. Let $c(i) = \sqrt{2log(\frac{1}{\delta})\frac{i}{s_2}}$. Then, conditioned on $\bm{X}$ for any top-k policy $\pi^k$ obtain as an output of algorithm \ref{algo_1}, the expected value of its certificate $f(\pi^k)$  is bounded below by 
\small
\begin{align}
    f(\pi^*) 
    &- \sum_{i = 1}^{k^*} 
    \exp\left( \frac{ -\left[ \Delta_{k^* i} - \left( c(k^*) - c(i) \right) \right]^2 s_1 }{ n } \right ) \nonumber \\
    &\quad \times \left[ \Delta_{k^* i} - \left( c(k^*) - c(i) \right) \right] \nonumber \\
    &- \sum_{i = k^*+1}^n 
    \exp\left( \frac{ -\Delta_{i k^*}^2 s_1 }{ n } \right )\left[ \Delta_{k^* i} - \left( c(k^*) - c(i) \right) \right]
\end{align}

\end{restatable}

Where $\Delta_{ij} = \mu_{\sigma_i} - \mu_{\sigma_j}$.  This result is very similar to the classical simple regret guarantees \cite{bubeck2011pure}, except gaps between means are replaced with gaps between lower bounds. Notably, although the bound depends on $k^*$ our algorithm (and thus the guarantee) does not, making our procedure adaptive. In particular, if $k^* = n$, the bound will be identical to the one obtained from a single stage approach. Notably, as a consequence of the gaps being part of the bound, the tension between $s_1$ and $s_2$ becomes explicit; higher $s_1$ leads to exponential improvements on the bounds at the cost of the loosening it by $\frac{1}{\sqrt{s_2}}$.

\subsection{Incorporating Priors}


So far, we have developed algorithms which are prior-free: they do not rely on any information in advance about $\mu$ and aim to approximate the optimal policy purely using information gathered during the trial. However, in many practical settings, the experimenter may have an informative prior over $\bm{\mu}$ given previous work in the domain. For example, many meta-analyses have quantified the distribution of reported effect sizes for interventions in domains such as education~\citep{education_effect}, medicine~\citep{real_world_study}, development~\citep{economic_effect}, and more. Although the effect size for new interventions under consideration are unknown, the experimenter may be able to improve their design by modeling them as drawn from such a prior distribution.  

Formally, we assume access to a joint prior, $\mathcal{P}$, so that the mean of each arm is distributed according to $\mu_{1},\mu_{2},\ldots,\mu_{n} \sim \mathcal{P}$. Our goal then, is to find a policy which maximizes the previously defined certificate when arms are distributed according to the prior, that is finding the policy $\pi^*$ that maximizes:
\begin{equation} 
\small
f(\pi) = \mathbb{E}_{\bm{\mu} \sim \mathcal{P}}\left[\mathbb{E}_{\mathbf{X}, \mathbf{Y} \sim \bm{\mu}} \left[\max\limits_{i \in \pi(\mathbf{X})} l_{i}(\mathbf{Y})\right]\right]
\end{equation}

To find the optimal policy, we note that $\pi^*$ can be computed for each $\mathbf{X}$ by optimizing over the posterior $\mathcal{P}(\bm{\mu} | \mathbf{X})$, in other words by maximizing,
\begin{equation}
    \small
   f_{\mathbf{X}}(\pi) = \mathbb{E}_{\bm{\mu} \sim \mathcal{P}(\bm{\mu} | \mathbf{X})} \left[\mathbb{E}_{\mathbf{Y} \sim \bm{\mu}}\left[\max\limits_{i \in \pi(\mathbf{X})} l_{i}(\mathbf{Y}) | \mathbf{X} \right]\right] 
\end{equation}
for all $\mathbf{X}$. Importantly, this certificate $l$ still has the same frequentist coverage guarantees; the Bayesian prior is used only to improve the power of the design.

We develop a prior-based algorithm to greedily compute $\pi(\mathbf{X})$, assuming sampling access to the posterior distribution. 
The sampling can be implemented through many different Bayesian inference approaches; in this paper, we ignore the specifics, and let this be a black box. The procedure begins by sampling $d$ values of $\bm{\mu}$ from the posterior. Then, for every set $B^1 \subset [n]$ of size 1, using the sample-average certificate obtained from the posterior draw, we estimate the value of $\mathbb{E}_{\mathbf{Y} \sim \bm{\mu}}\left[\max\limits_{i \in B^1} l_{i}(\mathbf{Y})\right]$. 
This process is repeated for all possible sets of size two, $B^1 \cup \{i\}$ for $i \in [n] \setminus B^1$, to greedily construct $B^2$, a candidate set of size two. 
The algorithm proceeds iteratively, adding elements to $B^i$ until $|B^i| = n$. 
Finally, the resulting policy $\pi(\mathbf{X})$, is defined as the set $B^k$ with the highest estimate out of $B^1...B^n$. Consitent with the previous sections let $B^k = \pi(\mathbf{X})$

Through reduction to submodular optimization, we demonstrate approximation guarantees for our algorithm. 
Our algorithm relaxes the assumption of stochastic domination needed previously (In particular our derived policy is not a top-k policy). Formally,


\begin{restatable}{theorem}{lowerboundprior}
\label{thm:lowerboundprior}
    Let $\hat{\pi}$ be the policy obtained by our proposed algorithm using $d$ samples from the posterior $\mathcal{P}(\bm{\mu} | \mathbf{X})$. Then  $f(\hat{\pi}) \geq f(\pi^*) (1-1/e-\epsilon)$, where $\epsilon = O\left(\sqrt{\log\left( \frac{1}{\delta}\right){d}} \right)$.
\end{restatable}

The idea behind the proof is to leverage the property that greedy algorithms are $1-1/e$ optimal for monotonic submodular optimization, and we demonstrate that our situation does in fact match this type of problem. Furthermore, we demonstrate a hardness result, showing that we cannot beat the $1-1/e$ performance achieved by our posterior sampling algorithm \ref{thm:upperboundprior} the proof of this result can be found on the appendix.

\section{Experiments}
\label{sec:experiments}
We assess our two-stage RCT design with both synthetic and real-world datasets.

\begin{figure*}[ht!]
\centering 
\includegraphics[width=\linewidth]{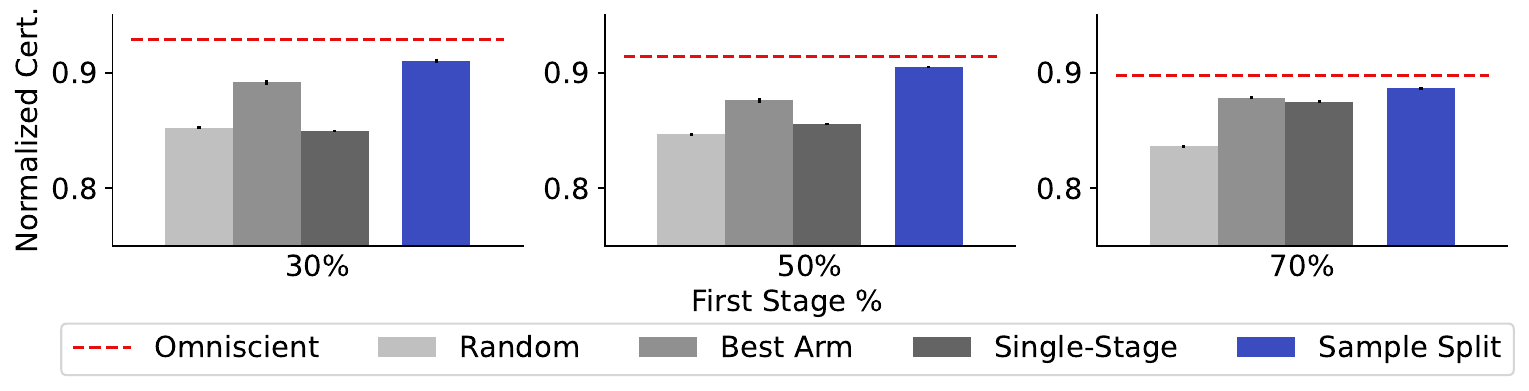}
\caption{Our sample splitting algorithm outperforms all baselines across first stage sizes. The largest improvement occurs when the first stage is small, as this leaves budget for the second stage to compute certificates. }
\label{fig:non_adaptive}
\end{figure*}

\paragraph{Synthetic Dataset and Setup} 
We construct a synthetic dataset to evaluate our two-stage RCT designs. 
We sample arm means, $\bm{\mu}$, from a uniform 0-1 distribution (we experiment with other choices  in Appendix~\ref{sec:uniform}). 
Arms have Bernoulli outcomes with mean $\mu_{i}$, which simulates settings where treatment are successful with probability $\mu_{i}$. 
We fix $n=10$ (we find similar results for other $n$ in Appendix~\ref{sec:n_arms}) and $\delta=0.1$ (and find similar results for other $\delta$ in Appendix~\ref{sec:delta}). 
We compare the following RCT designs
\begin{enumerate}
    \item \textbf{Random} - Two-stage top-k method + random $k$
    \item \textbf{Best Arm} - Two-stage method with $k=1$
    \item \textbf{single-stage} - the basic RCT which uniformly randomizes the entire budget $T$ over the arms, without using a second stage to eliminate some.
    \item \textbf{Sample Split} - Our proposed two-stage method which uses the first stage to prune arms and the second stage to compute certificates
    \item \textbf{Omniscient} - We construct an omniscient top-k policy. It consists on delivering the top $k^*$ arms from the first stage where $k^*$ is optimal. In other words, $k^* = \argmax_{i \in \pi(\bf{X})} \mu_i - \sqrt{log(\frac{1}{\delta})\frac{i}{2 s_2}}$.
\end{enumerate}

We compare designs by measuring normalized certificates, which is the ratio of $l$ to $\max\limits_{i} \mu_{i}$. 
We average results over $15$ seeds and $100$ runs per seed; seeds sample values of $\bm{\mu}$, while runs sample values for $\mathbf{X}$ and $\mathbf{Y}$. 

\paragraph{Comparison against Single-Stage Designs}

We compare our two-stage design against baselines, and find that our sample splitting methods improve upon baselines. 
In Figure~\ref{fig:non_adaptive}, we find that our sample splitting methods outperform single-stage methods across first stage percentages. 
When $s_{1}$ is 30\% of the budget, sample splitting methods outperform single-stage methods by 7\%, while when $s_{1}$ is 70\% of the budget, sample splitting outperform single-stage methods by only 1\%. 
With $s_{1}$ is large, $s_{2}$ is smaller, and so sampling splitting algorithms have less data to use for calculating certificates. 
However, when $s_{1}$ and $s_{2}$ are balanced, we can prune many arms in the first stage, while leaving sufficient time to find certificate values.

\paragraph{Impact of Budget}

\begin{figure}
    \includegraphics[width=\linewidth]{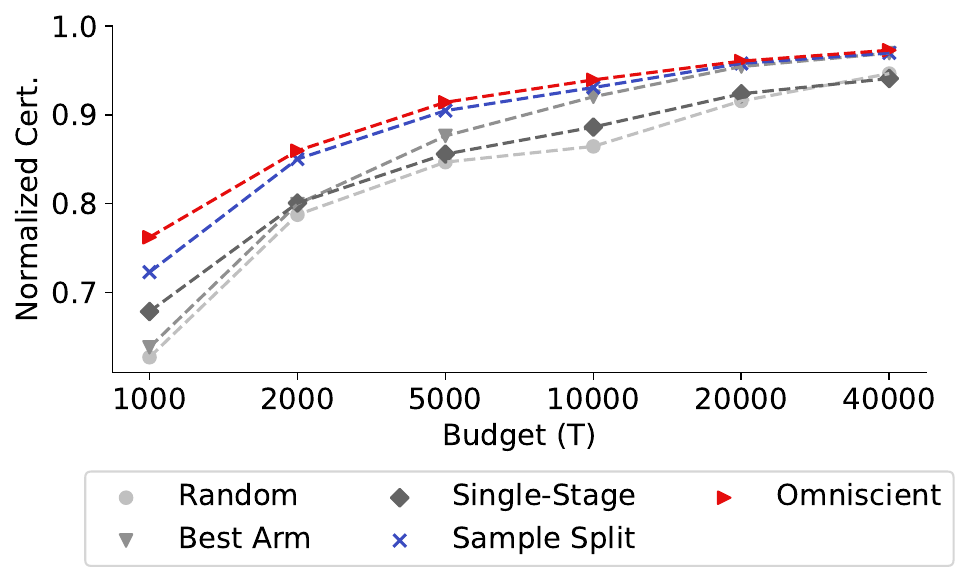}
    \caption{Sample split algorithms perform well for all values of $T$. When $T$ is large, the sample split approaches the optimal two-stage policy (omniscient). }
    \label{fig:vary_budget}
\end{figure}

To understand our designs across choices of $s_{1}$ and $T$, we compare policy performance, both when a) letting $s_{1} = s_{2}$, while varying $T$, and b) fixing $T$ while varying $s_{1}$. 

In Figure~\ref{fig:vary_budget}, we find that sample splitting policies are significantly better than all baselines when $T \leq 10000$ ($p<10^{-11}$). 
While best arm policies are 3\% better than sample splitting policies for $T=40000$, best arm policies are 13\% worse than sample splitting for $T=1000$. 
When comparing against the omniscient certificate, we see that sample splitting policies approach the omniscient policy for large budgets, as they are within 0.5\% for $T=40000$. 

\begin{figure}
\centering 
\includegraphics[width=\linewidth]{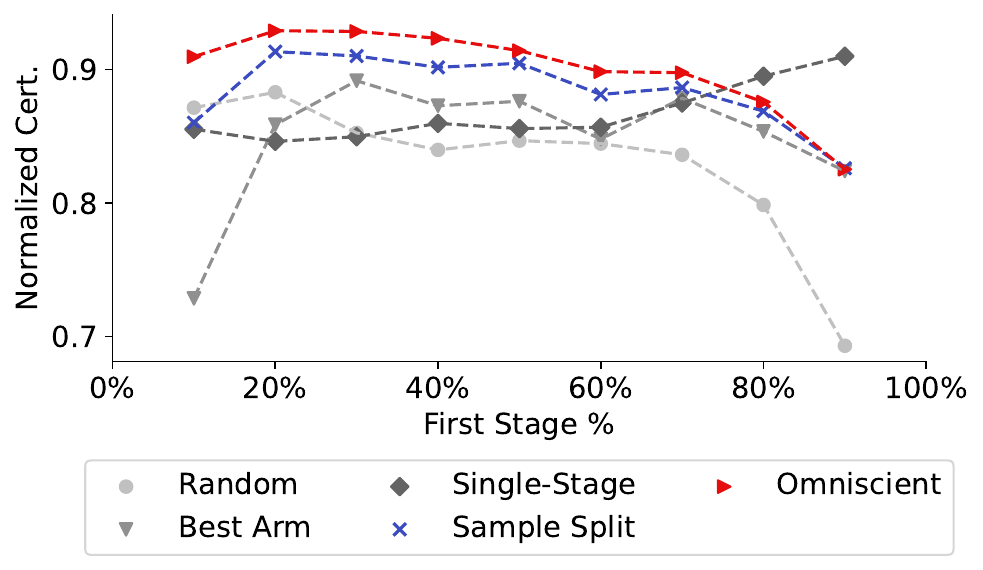}
\caption{Single-stage methods perform best when between 20\% and 70\% of the budget spent in the first stage, as this allows for arms to be pruned, and a certificate to be generated in the second stage. }
\label{fig:vary_stage_percent}
\end{figure}

\begin{figure*}[h]
\centering 
\includegraphics[width=\linewidth]{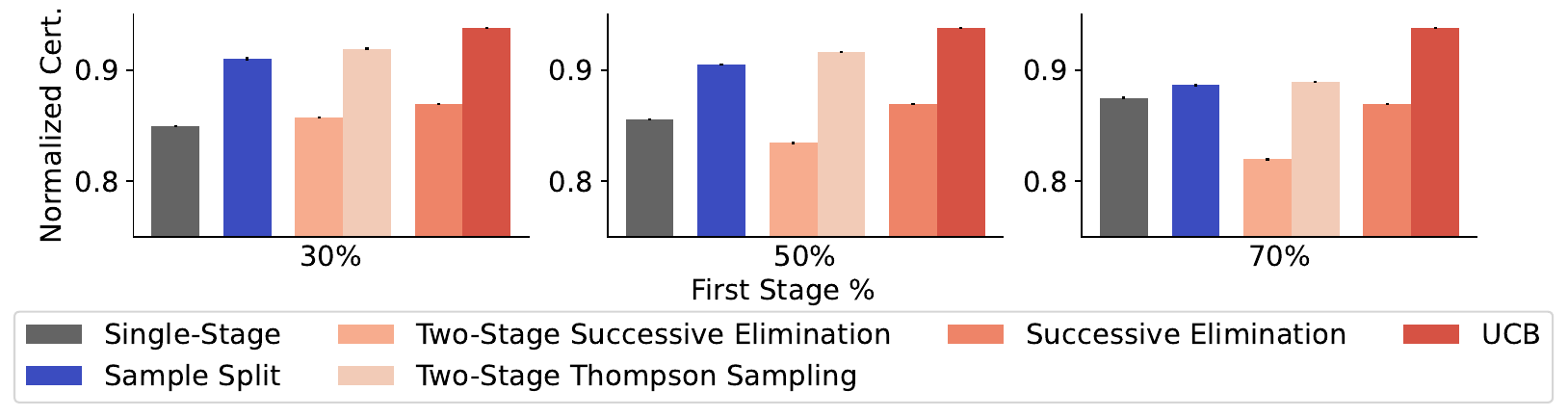}

\caption{Sample splitting policies can close the gap between single-stage and adaptive policies (such as UCB), and serves as a middle ground in performance and complexity. This is best seen when the first stage is 30\%, as sample split methods can capture up to 69\% of the improvement between single-stage and UCB designs. }
\label{fig:two_stage}
\end{figure*}

We compare our two-stage methods against baselines when varying the ratio of $s_{1}$ to $T$. 
In Figure~\ref{fig:vary_stage_percent}, we show that sample splitting excels with large first-stage sizes, whereas single-stage methods perform best when $\frac{s_{1}}{T} \geq 70\%$. 
When the first stage is much larger than the second stage, two-stage methods uses only data from the second stage, while single-stage methods use data from both stages.
This occurs due to the adaptivity from the second stage, which prevent us from applying Hoeffdings bound to data from both stages due to the lack of the iid property.

\paragraph{Comparison with Adaptive Policies}

We compare our approach against  adaptive designs which can potentially improve the guarantees at the cost of complexity. 
We detail the adaptive methods below: 
\begin{enumerate}    
    \item \textbf{Two-Stage Successive Elimination} - We perform successive elimination in-batch, by first doing uniform exploration in the first stage, then running successive elimination to prune arms~\citep{successive_elimination}. We then re-run uniform exploration in the second stage. 
    \item \textbf{Two-Stage Thompson Sampling} - 
    In the first stage, two-stage Thompson Sampling performs performs uniform exploration, and in the second stage, it performs non-uniform exploration in the second round based on probabilities from  Thompson sampling probabilities with a uniform prior~\citep{thompson_sampling}. 
    \item \textbf{Successive Elimination} - We run the successive elimination algorithm, with a budget of $T$. 
    \item \textbf{UCB} - We run the upper confidence bound (UCB) algorithm with budget $T$~\citep{auer2002finite}. Note that UCB has a much higher degree of adaptivity: it updates assignment probabilities $T-1$ times, compared to once for the two-stage policies.
\end{enumerate}

We compare our design to adaptive approaches in Figure~\ref{fig:two_stage}.
We find that non-uniform allocation probabilities have limited benefit: 
Sample splitting performs within 1.5\% of two-stage Thompson Sampling. Sample splitting also performs significantly better than both the two-stage and fully adaptive versions of Successive Elimination. 
UCB is the strongest fully adaptive algorithm and is able to perform better than less-adaptive designs. However, the best sample split design (where $s_1$ is 30\% of the budget) comes relatively close: Sample Splitting performs within 3\% of UCB, capturing up to 69\% of the improvement between the single-stage design and UCB. Practioners can capture much of the value of the most complex, highly adaptive design by using a properly configured two-stage, framework.   

\paragraph{Bayesian Setting} 
We explore whether knowledge of a prior distribution can improve the certificates discovered by two-stage policies. 
We compare our prior-based policy against both sample splitting and baseline policies on a synthetic dataset.
To construct such a dataset, we let $\bm{\mu}$ be distributed according to a $\beta$ distribution, fixing $\alpha=1$ and varying $\beta \in \{1,2,4\}$. Higher $\beta$ indicates a more informative prior. 

\begin{figure*}
\centering 
\includegraphics[width=\linewidth]{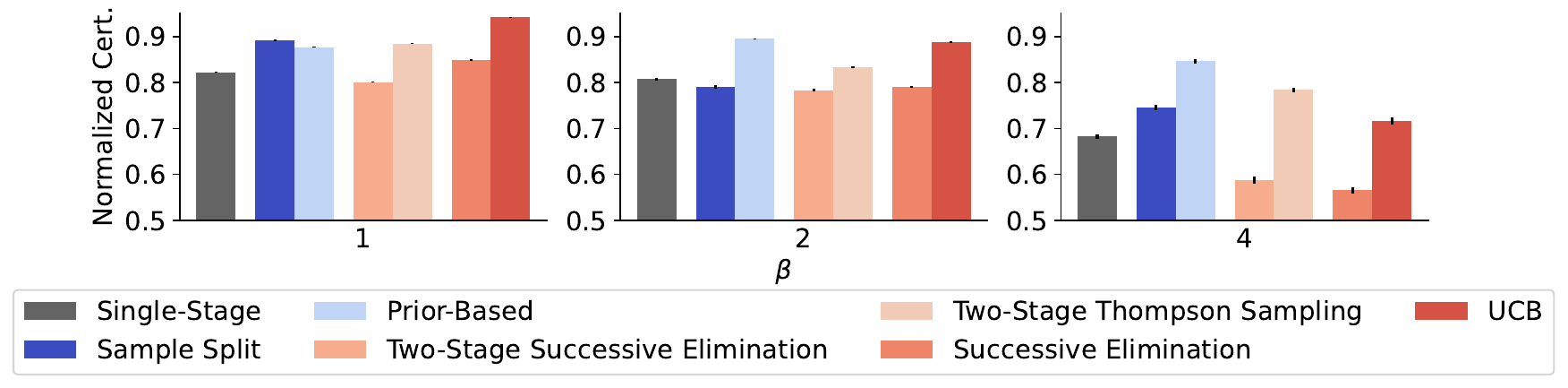}
\caption{When priors are informative, correlating to large $\beta$, prior-based methods can improve upon all policies, including adaptive policies such as UCB and two-stage Thompson Sampling. }
\label{fig:prior}
\end{figure*}

Figure~\ref{fig:prior} shows that informative priors (large $\beta$ allows prior-based methods to perform well, as they slightly exceed the performance of UCB at $\beta = 2$ and improve upon UCB by 18\% at $\beta = 4$. 
When available, informative priors of effect sizes contribute more than the ability to incorporate high degrees of adaptivity. 

\begin{figure}[h]
\centering 
\includegraphics[width=\linewidth]{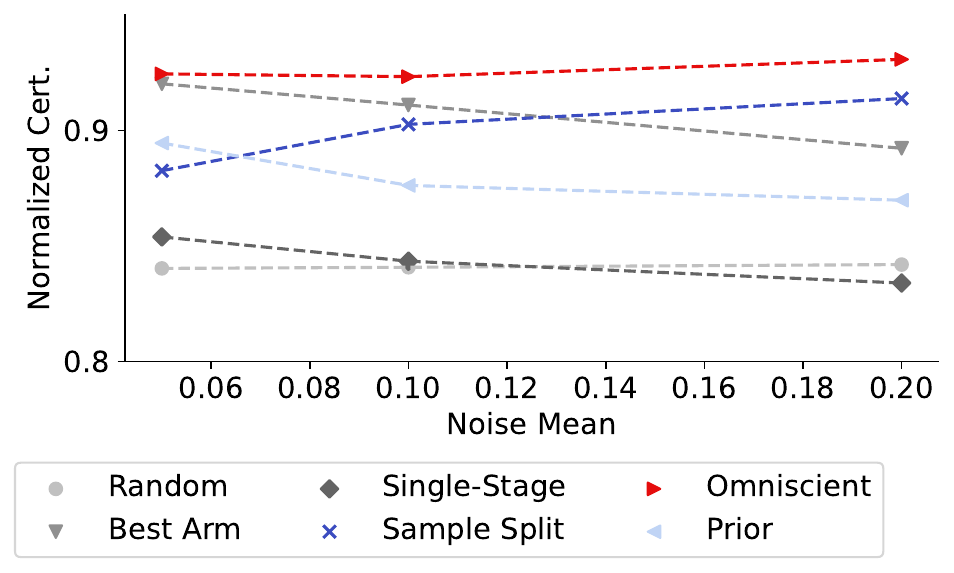}
\caption{Prior-based policies are sensitive to the mean of the noise; directional noise (corresponding to larger values for noise) leads to degrading performance, especially compared with sample splitting algorithms. }
\label{fig:misspecification}
\end{figure}

We compare the performance of our policies under prior misspecification by adding Gaussian noise to $\bm{\mu}$ with $\alpha=\beta=1$.  
We fix the noise variance to be $0.01$ and vary the mean in $\{0.05,0.1,0.2\}$. 
In Figure~\ref{fig:misspecification}, we see that when the prior is minimally misspecified, prior-based algorithms outperform sample split. 
However, prior-based algrithms fare poorly with increasing misspecification, demonstrating that prior-based are not robust to large degrees of misspecifcation. 

\paragraph{Real-World Experiments}
\label{sec:real_world}
\begin{figure}[h]
\centering 
\includegraphics[width=\linewidth]{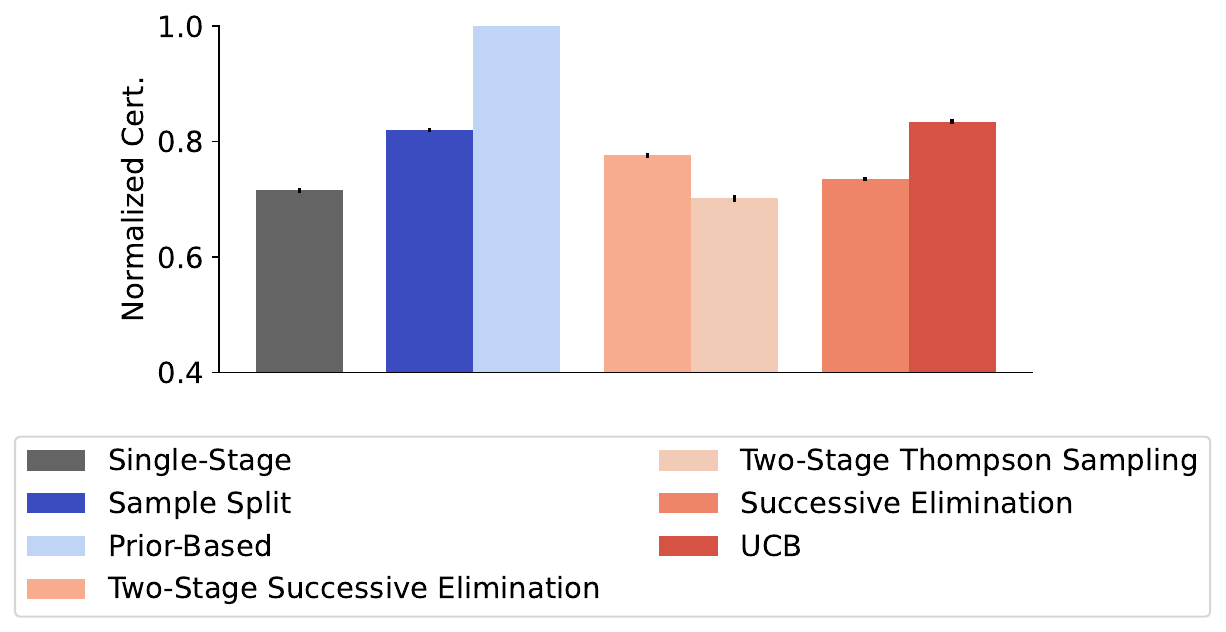}
\caption{On a real-world genertology dataset, we see that using prior-based methods improves the certificates generated by RCTs, even compared to adaptive methods. This reflects the added benefit we get when using domain knowledge for real-world RCTs.}
\label{fig:real_world}
\end{figure}

We run semi-synthetic experiments where effect sizes are drawn accordingly to a real-world distribution drawn from a meta-analysis of treatments in gerontology~\citep{real_world_study}. We retrieve 75 effect sizes fromthe meta-analysis and set the prior on $\bm{\mu}$ to be uniform over these values. 
Since effect sizes are reported as Cohen's d (a standardized metric), we model the outcome distribution $D_{\mu_i}$ as a normal with mean $\mu_i$ and standard deviation 1. 
For the certificate $\ell$, we use the 
corresponding tail bound for Sub-Gaussian variables: 

\begin{equation}    
\small
    P\left(|\bar{Y}_{i} - \mu_{i}| \leq \sqrt{\frac{2k}{s_{2}} \log(\frac{2}{\delta})}\right) \geq 1-\delta 
\end{equation}

In Figure~\ref{fig:real_world}, we see that prior-based methods perform best, beating even adaptive methods like UCB by 23\%. 
Restricting our attention to prior-free algorithms, sample split performs within 2\% of UCB. 
This verifies that the main conclusions from the synthetic experiments continue to hold on a real-world distribution: two-stage designs based on our sample split procedure can nearly match the performance of fully adaptive designs, and with access to the prior, Bayesian two-stage designs perform significantly better. 

\section{Conclusion and Limitations}
Traditional single-stage RCTs spend unnecessary resources exploring sub-optimal arms, while fully-adaptive procedures are often costly. 
We study two-stage RCTs, which improves on the guarantees from single-stage RCTs, while deliberately maintaining simplicity. 
We develop a top-k algorithm for designing such RCTs, and demonstrate the optimality of top-k two-stage designs when distributions posses a stochastic dominance ordering. 
We empirically demonstrate that our two-stage RCT can significantly improve guarantees compared with single-stage RCTs, and can even outperform adaptive methods in Bayesian settings. 
By using two-stage RCTs, real-world studies can improve guarantees without increasing complexity. 

Real-world constraints should inform the design of our two-stage RCTs. 
For example, our work incorporated domain knowledge through the use of priors, and future work could look into other ways that domain knowledge could be incorporated. 
Better understanding the level of adaptivity available can also instruct us how multi-stage RCTs can be better designed. 
For example, with a larger adaptivity budget can allow us to construct three or four-stage RCTs, which could potentially improve the guarantees delivered by our algorithms. 
Our work demonstrates the added benefits when introducing two-stage RCTs, and these methods can be customized depending on the situation.

\bibliographystyle{unsrtnat}
\bibliography{ref}

\section{Certificate Spread}
\begin{figure*}[h]
\centering 
\includegraphics[width=\linewidth]{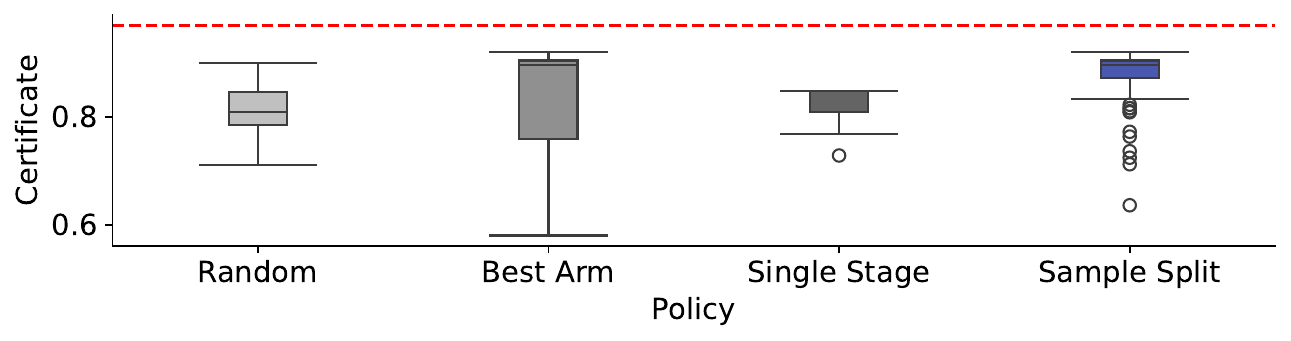}
\caption{We compare the spread in certificate value across various policies. We find that our two-stage sample splitting policy improves the average certificate value, while reducing the spread of certificate values. }
\label{fig:box_plot}
\end{figure*}
We plot the spread of the certificates in Figure~\ref{fig:box_plot}. 
We find that our sample splitting policy has a lower spread compared to single-stage, even with fixed $k=1$ (best arm). 
This demonstrates how our sample splitting algorithm can reduce the uncertainty in certificates produced. 
Here, the dashed red line refers to the actual true value. 

\section{Impact of Model Confidence}
\label{sec:delta}
\begin{figure*}[h]
\centering 
\includegraphics[width=\linewidth]{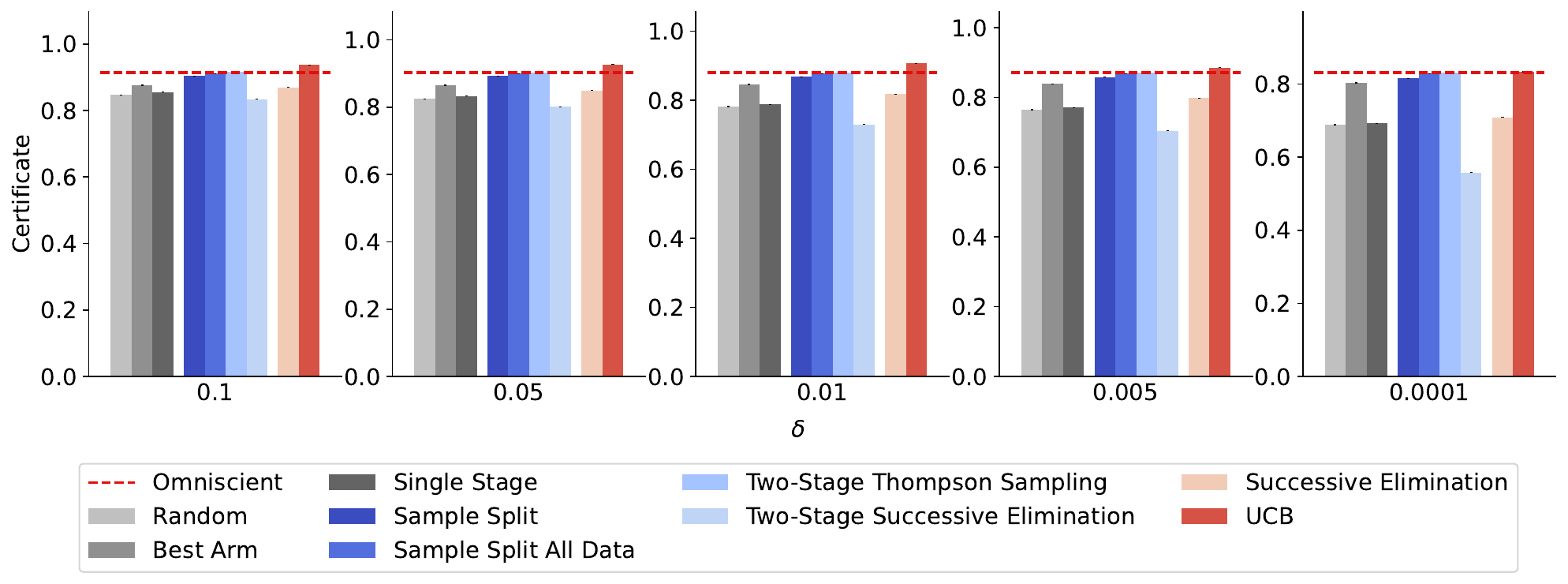}
\caption{Adaptive methods perform well even with increasing $\delta$, while the adaptivity gap increases with increasing $\delta$. }
\label{fig:delta}
\end{figure*}

Variations in confidence should naturally lead to looser certificates, yet the rate of looseness varies across methods.
We compare policy performance as we vary $\delta$. 
In Figure~\ref{fig:delta}, we see that for very small $\delta$, adaptive policies are as good as sample splitting policies. 
We additionally see that the value of $\delta$ has a minimal impact on the relative trend between two-stage algorithms against single-stage or random baselines.

\section{Additional Stages}

\begin{figure}[h]
\centering 
\includegraphics[width=\linewidth]{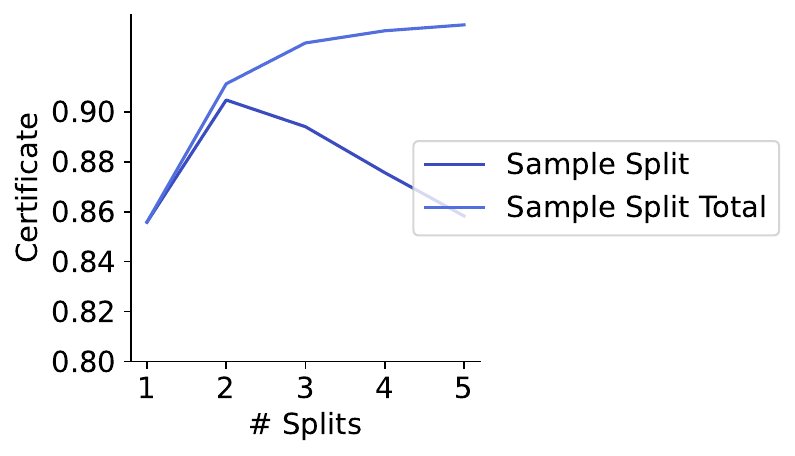}
\caption{We extend two-stage methods up to five stages, and compare the performance when using either only last stage data or all of the data. We find that two-stage is optimal when using only the last stage, while when using all of the data, we find that using more stages increases performance, though with diminishing improvements in certificate. }
\label{fig:multi_stage}
\end{figure}
We consider the impact of extending two-stage algorithms to multiple stages in Figure~\ref{fig:multi_stage}. 
We consider two settings for this; one where only the last stage is used for the certificate, and another where all stages are used. 
\nrcomment{Don't call it split}
We see that, when using only the last stage, the optimal number of stages is $2$. 
However, when increasing the number of stages and using all the data, we see that it monotonically increases, with diminishing outcomes. 
Such results justify our decision to study two-stage splitting, as we can recover much of the gains, especially when only using the last stage, by using a two-stage approach. 

\section{Varying the number of arms}
\label{sec:n_arms}
\begin{figure*}[h]
\centering 
\includegraphics[width=\linewidth]{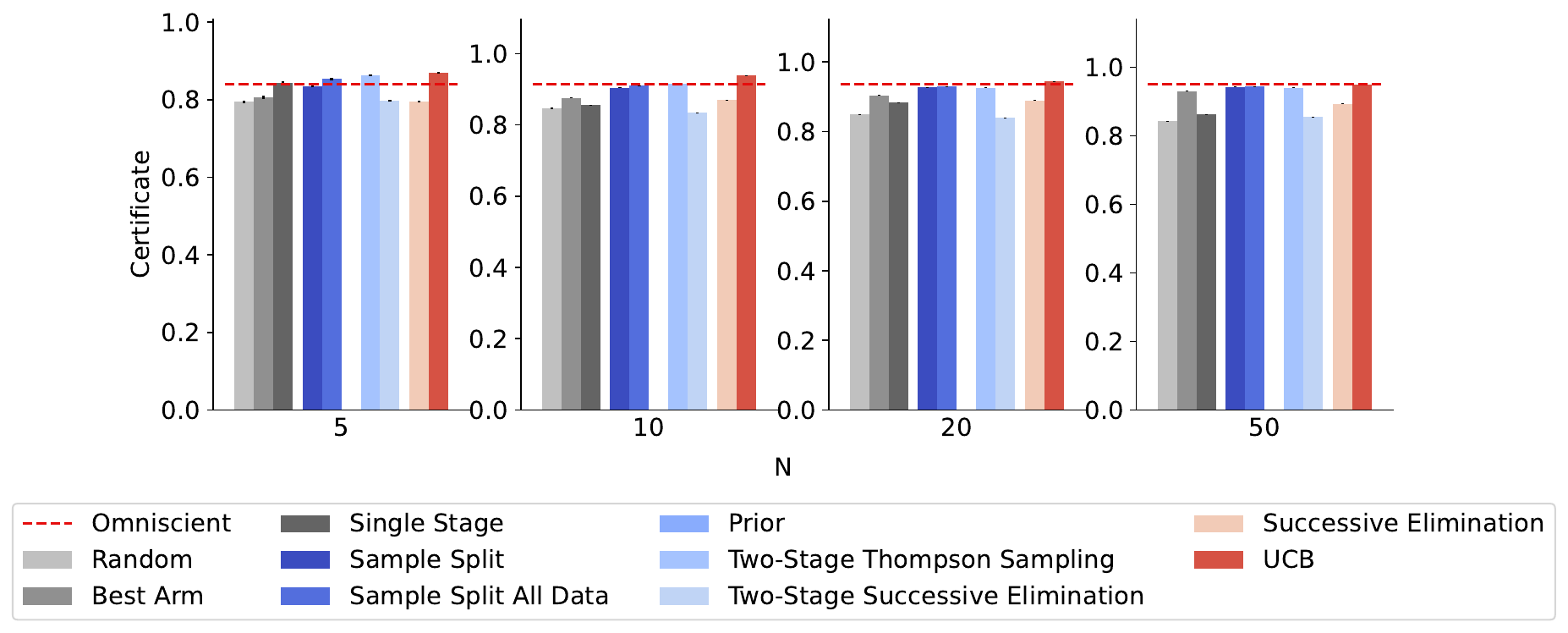}
\caption{We vary the number of arms, $n$, and find similar trends independent of the number of arms. }
\label{fig:n_arms}
\end{figure*}
We vary the parameter $n$ in Figure~\ref{fig:n_arms}. 
We find that across values of $n$, our policies perform similarly. 
Larger $n$ reduces the gap between adaptive and two-stage policies, similar to the gap with increasing $\delta$. 
This might be because more time is spent exploring with larger $n$ for adaptive policies. 
However, the same trends apply across values of $n$. 
We fainlly note that, for large $n$, the gap between single-stage and two-stage policies increases. 

\section{Varying distribution of arm means}
\label{sec:uniform}
\begin{figure*}[h]
\centering 
\includegraphics[width=\linewidth]{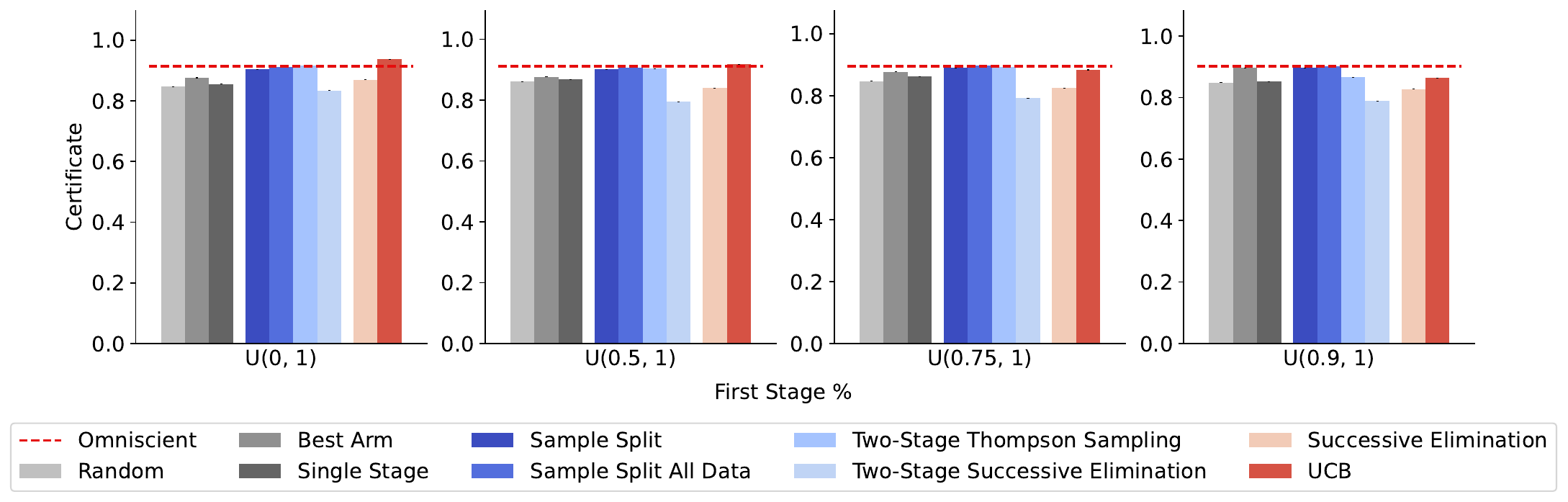}
\caption{We vary the distribution of the underlying arm means distribution, and find similar performance across distributions.}
\label{fig:distribution}
\end{figure*}
We vary the underlying distribution of arm means, $\mu$, as a Uniform distribution with varying parameters. 
\nrcomment{Change x label}
When arms are pushed near 1, then adaptive policies perform worse.
Here, the best policy is either top-1 or sample splitting. 
With smaller mean arms, we see that top-1 policies are worse, and so sample splitting is the best policy to modulate between the two.

\section{Proofs}
\label{sec:proofs}
\expectedvaluecertificate*
\begin{proof}

  As a consequence of Theorem 1 and the stochastic dominance of Bernoulli arms, it suffices to analyze the error against the top-$k^*$ policy. Let $\pi^k$ be the top $k$ policy outpuuted by algorithm \ref{algo_1}. Clearly
 \begin{align*}
     f(\pi^*) - f(\pi^k) =  \E[\max_{i \in \pi^*(X)}\mu_i - c(k^*)] - \E[\max_{j \in \pi^k(X)}\mu_j - c(k)]
 \end{align*}

 \textbf{Case 1}, $k < k^*$. In this scenario, $c(k^*) - c(k) \geq 0$, hence if $\max_{ j \in \pi^k(X)}\mu_j \geq \max_{ j \in \pi^*(X)}\mu_j$ that will contradict the optimality of $\pi^*$, furthermore, this is true for every $k < k^*$, thus concluding that $\mu_{\sigma_*(X)} > \mu_{\sigma_i(X)}$ and $\mu_{\sigma_*(X)} = \max_{i \in \pi^*}\mu_j$.  Additionally by optimality of $k^*$, $\mu_{\sigma_*(X)} - c(k^*) > \max_{ j \in \pi^k(X)}\mu_j - c(k)$. By definition of the alorithm \ref{algo_1}, $\overline{V}_{\sigma_k} - c(k)> \overline{V}_{\sigma_{k^*}} - c(k^*)$ then given that $c(k^*) - c(k) < \mu_{\sigma_*(X)} - \mu_k$. In particular:
 \begin{align*}
    &\leq \sum_{i = 1}^{k^*} P(\overline{V}_{\sigma_{k^*}} - \overline{V}_{\sigma_i} -c(k^*) + c(i) < 0)[\Delta_{k^*i} - ( c(k^*) - c(i))]\\
    &\leq \sum_{i = 1}^{k^*} exp(\frac{-2[\Delta_{k^*i}- (c(k^*) - c(i)) ]^2s_1}{2n})[\Delta_{k^*i} - (c(k^*) - c(k))]
 \end{align*}

 \textbf{Case 2} $k^* < k$, $c(k^*) - c(k) < 0$, thus as $\overline{V}_{k^*} - \overline{V}_k < c(k^*) - c(k)$ (Otherwise the index $k$ would not have been picked) then it must be concluded that $\overline{V}_{k^*} - \overline{V}_i < 0$. Then we can bound this terms using the event $P(\overline{V}_{k^*} - \overline{V}_k < 0)$ in which the traditional Hoeffding's bound trick can be plugged.

\begin{align*}
    &\leq \sum_{i = *+1}^n P(\overline{V}_{k^*} - \overline{V}_i < 0)[\Delta_{k^*i} - (c(k^*) - c(i))]\\
    &\leq \sum_{i = k^*+1}^n P(\overline{V}_{k^*} - \overline{V}_i - \Delta_{k^*i}< -\Delta_{k^*i})[\Delta_{k^*i} - (c(k^*) - c(i))]\\
    &\leq \sum_{i = *+1}^n exp(\frac{-2\Delta_{ik^*}^2s_1}{2n})[\Delta_{k^*i} - (c(k^*) - c(i))]
\end{align*}
Note that although $\Delta_{k^*i}$ might be negative, $\Delta_{k^*i} - (c(k^*) - c(i))$ is not as a consequence of the optimality of $k^*$.
\end{proof}

\greatertopk*
\begin{proof}
    We conclude that since $\max_{i \in \pi(X)}\{\mu_i\}$ is a non-decreasing function in arguments $\{\mu_i: i \in \pi(X)\}$, first-order stochastic dominance implies (Is equivalent to the expectations matching for every non decreasing function).
\begin{align*}
    \E[\max_{i \in \pi'(X)}\mu_i] \geq \E[\max_{i \in \pi(X)}\mu_i]
\end{align*}
from which it follows that $f_{\bm{\mu}}(\pi') \geq f_{\bm{\mu}}(\pi)$     
\end{proof}
\topk
\begin{proof}
    Let $\pi^*$ be the optimal policy, as a consequence of Lemma \ref{thm:top_k_dominance} then there exists $\pi'$, a top-k policy, such that $f(\pi') \geq f(\pi)$. By optimallity the of $\pi^*$ the result follows.
\end{proof}

\stochasticdomination*
\begin{proof}
    Fix a given $t$. We compare $\Pr(\mu_{\sigma_i(X)} \geq t)$ and $\Pr(\mu_{\sigma_j(X)} \geq t)$ via  a coupling argument where we fix the mean of every position in the permutation except for positions $i$ and $j$. Formally, we decompose
    \begin{align*}
        \Pr(\mu_{\sigma_i(X)} \geq t) - \Pr(\mu_{\sigma_j(X)} \geq t) =\\ \E[1[\mu_{\sigma_i(X)} \geq t] - 1[\mu_{\sigma_j(X)} \geq t]]\\
        = \E[\E[1[\mu_{\sigma_i(X)} \geq t] - 1[\mu_{\sigma_j(X)} \geq t]|\mu_{\sigma_k(X)}, k \neq i,j]].
    \end{align*}
    We will show that the inner expectation $\E[1[\mu_{\sigma_i(X)} \geq t] - 1[\mu_{\sigma_j(X)} \geq t]|\mu_{\sigma_k(X)}, k \neq i,j]$ is always nonnegative. After conditioning on $\mu_{\sigma_k(X)}, k \neq i,j$, there are two possible values for $\mu_{\sigma_i(X)}$ and $\mu_{\sigma_j(X)}$ (i.e., the means of the two arms that did not appear in other positions of the permutation). Denote these values by $\mu_a$ and $\mu_b$ where $\mu_a \geq \mu_b$ refers to the larger of the two. There are two cases. The first case is where either $\mu_a \geq \mu_b \geq t$ or $t \geq \mu_a \geq \mu_b$. Denote this event as $\mathcal{E}_1$. By definition,   $\E[1[\mu_{\sigma_i(X)} \geq t] - 1[\mu_{\sigma_j(X)} \geq t]|\mathcal{E}_1, \mu_{\sigma_k(X)}, k \neq i,j] = 0$. The second case is that $\mu_a \geq t > \mu_b$. Denote this event as $\mathcal{E}_2$. Since $i < j$, $\sigma_i(X) = a$ if and only if $\overline{X}_a > \overline{X}_b$, which in turn implies that  $\mu_{\sigma_i(X)} \geq t$ if and only if either $\overline{X}_a > \overline{X}_b$ or  $\overline{X}_a = \overline{X}_b$ and $r_a > r_b$.  Because $\mu_a \geq \mu_b$, the assumption of a stochastic dominance ordering on the arms combined with independence of the samples implies that  $Pr(\overline{X}_a > \overline{X}_b|\mathcal{E}_2, \mu_{\sigma_k(X)}, k \neq i,j) \geq Pr(\overline{X}_b > \overline{X}_a|\mathcal{E}_2, \mu_{\sigma_k(X)}, k \neq i,j) $. Moreoever, whenever $\overline{X}_a = \overline{X}_b$, $\Pr(r_a > r_b) = \frac{1}{2}$ by definition.  Since conditioned on $\mathcal{E}_2$, $1[\mu_{\sigma_j(X)} \geq t] = 0$ whenever $1[\mu_{\sigma_i(X)} \geq t] = 1$, we conclude that $\E[1[\mu_{\sigma_i(X)} \geq t] - 1[\mu_{\sigma_j(X)} \geq t]|\mathcal{E}_2, \mu_{\sigma_k(X)}, k \neq i,j] =\geq 0$ and the lemma follows.
\end{proof}

\lowerboundprior*
\begin{proof}
    We aim to show that  $\hat{\pi}$ is such that  $f_{\bm{\mu}}(\hat{\pi}) \geq f_{\bm{\mu}}(\pi^*) (1-1/e)$. 
    We do so by reducing our problem to an instance of cardinality constrained submodular optimization, and then note that we can solve such a problem up to a $1-1/e$ approximation in polynomial time~\citep{submodular_optimization}. 

    Let $B^i$ the set of size $i$ obtained from the algorithm. Let $\hat{Z}_{1},\hat{Z}_{2},\ldots,\hat{Z}_{d}$  be $d$ samples of the second stage pulls from the posterior. We use slight abuse of notation, and let $\hat{\mu}(Z_{j})$ be the set of empirical means, we then can approximate $f$ as  $(\frac{1}{d} \sum_{j=1}^{d} \max\limits_{i \in B} \hat{\mu}(Z_{j})) - O(\sqrt{\frac{i}{s_{2}}})$ with high probability. 
    Note that $O(\sqrt{\frac{i \log(\frac{1}{\delta})}{s_{2}}})$ turns our empirical mean into a high probability certificate. 
    We note that our interior optimization problem, $\frac{1}{d} \sum_{j=1}^{d} \max\limits_{i \in B} \hat{\mu}(Z_{j})$ is a sum of monotonic submodular functions (the maximum function), and therefore, our entire function is monotonic submodular. 
    
    To allow for our whole optimization problem to be monotonic and submodular, we consider all values of $|B^i| \leq n$, and can solve $n$ such submodular optimization problems. 
    We note that the outer function, $- O(\sqrt{\frac{i}{s_{2}}})$ is constant given $i$. 
    We therefore brute force all $n$ values of $i$, and each such submodular optimization can be solved in $O(nd)$. 
    Our overall problem can therefore be solved in $O(n^{2}d)$ to an approximation ratio of $1-1/e-\epsilon$ with high probaility, where $\epsilon = O(\sqrt{\log(\frac{1}{\delta}){d}})$.

\end{proof}

\begin{theorem}
\label{thm:upperboundprior}
    Consider the problem of identifying a high probability certificate using a two-stage sample splitting method. 
    Let there be $n$ arms, with the mean outcome for arm $i$, denoted $\mu_{i}$, distributed according to a joint prior $\mathcal{P}$. 
    Next, consider the problem of selecting an optimal subset of arms for  the second stage, i.e, finding $\pi^*$. Then no polynomial time algorithm can find a policy $\hat{\pi}$, so that $f(\hat{\pi}) > f(\pi^*)(1-1/e)$ for all priors $\mathcal{P}$. 
\end{theorem}
\begin{proof}
    We first sketch a relationship between this problem and the cardinality constrained instance. 
    Consider the problem of $\max\limits_{|\pi(X)| \leq r} \mathbb{E}[\max\limits_{i \in \pi(X)} l_{i}]$ for some fixed $r$. 
    By assumption, for fixed $k = |\pi(X)|$, we have that $l_{i} = X_{i}-g(k)$ for some function $g$ that does not dependent on $i$. 
    Any polynomial time algorithm that can solve our cardinality constrained problem can also solve the non-cardinality constrained problem through an additional factor of $n$ by searching through all values for the budget. 

    We next demonstrate that no polynomial time algorithm achieves better than a $1-1/e$ approximation to the problem of $\max\limits_{|\pi(X)| \leq r} \mathbb{E}[\max\limits_{i \in \pi(X)} X_{i}]$ when given access to the prior distribution $\mathcal{P}$.
    We do so by reducing the maximum coverage problem to this. 

    We consider the maximum coverage problem with sets $C_{1},C_{2},\ldots,C_{n}$, where each set has elements taken from a universe $U$ of size $d$. 
    Our goal is to produce a set of sets, $C'$, with at most $r$ sets, that maximizes: 
    \begin{equation}
         \left| \bigcup_{C_i \in C'}{S_i} \right|
    \end{equation}
    Alternatively, our objective can also be viewed as the number of elements in $U$ covered by the sets $C$; that is, we can write the objective as 
    \begin{equation}
        \sum_{i=1}^{d} \mathbbm{1}[\exists j, C_{j}=1, u_{i} \in C_{j}]
    \end{equation}

    To perform a reduction to our cardinality-constrained certificate problem, we consider the problem of selecting at most $r$ arms out of $n$ for second-stage evaluation. 
    We let $s_{1}=0$, and only consider a second-stage evaluation. 
    In this situation, the posterior is the prior, $\mathcal{P}$, so we select arms to maximize: 
    \begin{equation}
        \max\limits_{|\pi(X)| \leq r} \mathbb{E}_{\mathcal{P}}[\max\limits_{i} X_{i}]
    \end{equation}
    The prior distribution, $\mathcal{P}$, denotes the probability of selecting particular outcomes, which correspond to the vector $\mu$. 
    We let the prior in this situation be uniformly distributed according to $d$ vectors, $v_{1}, v_{2}, \ldots, v_{d}, v_{i} \in [0,1]^{n}$. 
    Here, let $d_{i,j} = 1$ if $U_{i} \in C_{j}$, where $U_{i}$ is the ith element in the universe $U$, and $C_{j}$ is the jth set $C$. 
    Because the expectation is uniform and taken over $d$ elements, we rewrite the maximization problem as follows: 
    \begin{align*}
        \max\limits_{|\pi(X)| \leq r} \mathbb{E}_{\mathcal{P}}[\max\limits_{i} X_{i}] = \max\limits_{\mathbf{b} \in \{0,1\}^{n}, \| b \| \leq r} 
        \frac{1}{d} \sum_{i=1}^{d}  \max\limits_{j} b_{j} d_{i,j} \\
    \end{align*}

    Because both $b_{j}$ and $d_{i,j}$ are 0-1 values, our objective can also be seen as 
    \begin{equation}
        \max\limits_{j} b_{j} d_{i,j} = \mathbbm{1}[\exists j, b_{j} = 1 \wedge d_{i,j} = 1]
    \end{equation}
    Because our objective is to find the largest such $B$, we drop the constant $\frac{1}{d}$, and rewrite the objective as 
    \begin{equation}
        \max\limits_{B} \sum_{i=1}^{d} \mathbbm{1}[\exists j, b_{j} = 1 \wedge d_{i,j} = 1]
    \end{equation}
    This objective matches our objective for the coverage problem; if we view $d$ as the matrix denoting which elements $i$ are contained in which sets $j$, then the objective for our certificate problem is the same as the objective for the coverage problem, just off by $\frac{1}{d}$. 
    Therefore, if we can solve the certificate problem for this instance, then we can solve the coverage problem.
    However, because the coverage problem has a 1-1/e lower bound, our cardinality-constrained certificate problem also has a 1-1/e lower bound~\citep{set_cover}. 

\end{proof}
\end{document}